\documentclass[11pt]{article}       

\usepackage{amsfonts,amsthm,amsmath,amssymb,euscript,mathrsfs}
\usepackage{bm}
\usepackage{color}
\usepackage{dcolumn}
\usepackage[shortlabels]{enumitem} 
\usepackage{geometry}
\usepackage{graphicx}
\usepackage{setspace} 
\usepackage[authoryear]{natbib}
\usepackage{hyperref}

\def\bbE{\mathbb{E}}

\def\bbI{\mathbb{I}}

\def\bbP{\mathbb{P}}

\def\bbR{\mathbb{R}}

\def\cD{{\mathcal D}}

\def\cL{{\mathcal L}}

\def\cO{{\mathcal O}}

\def\gd{\delta}

\def\gs{\sigma}

\def\gD{\Delta}

\def\gL{\Lambda}

\def\bfmath#1{\bm{#1}}

\def\bfi{{\bfmath{i}}}
\def\bfj{{\bfmath{j}}}
\def\bfk{{\bfmath{k}}}
\def\bfl{{\bfmath{l}}}

\def\tA{\theta_\A}
\def\tB{\theta_\B}
\def\A{{\mbox{\tiny $A$}}} 
\def\B{{\mbox{\tiny $B$}}}

\newcommand{\dd}{\mathop{}\!\mathrm{d}}
\def\ds{\ensuremath{\displaystyle}}
\DeclareMathOperator{\var}{Var}

\newtheorem{corollary}{Corollary}

\newtheorem{proposition}{Proposition}
\newtheorem{theorem}{Theorem}

\newcommand{\aref}[1]{Appendix~\ref{#1}}
\newcommand{\sref}[1]{Section~\ref{#1}}
\newcommand{\tref}[1]{Table~\ref{#1}}
\newcommand{\fref}[1]{Figure~\ref{#1}}
\newcommand{\cref}[1]{Chapter~\ref{#1}}

\newcommand{\propref}[1]{Proposition~\ref{#1}}
\newcommand{\thmref}[1]{Theorem~\ref{#1}}

\title{\vspace{-40pt}An estimator for the recombination rate from a continuously observed diffusion of haplotype frequencies}
\author{Robert C.~Griffiths\thanks{School of Mathematics, Monash University, 9 Rainforest Walk, Victoria 3800 Australia}~~and Paul A.~Jenkins\thanks{Department of Statistics, University of Warwick, Coventry CV4 7AL, UK} \thanks{Department of Computer Science, University of Warwick, Coventry CV4 7AL, UK} \thanks{The Alan Turing Institute, British Library, London NW1 2DB, UK}}
\date{\today}

\begin{document}

\maketitle

\singlespacing
\begin{abstract}
\noindent Recombination is a fundamental evolutionary force, but it is difficult to quantify because the effect of a recombination event on patterns of variation in a sample of genetic data can be hard to discern. Estimators for the recombination rate, which are usually based on the idea of integrating over the unobserved possible evolutionary histories of a sample, can therefore be noisy. Here we consider a related question: how would an estimator behave if the evolutionary history actually was observed? This would offer an upper bound on the performance of estimators used in practice. In this paper we derive an expression for the maximum likelihood estimator for the recombination rate based on a continuously observed, multi-locus, Wright--Fisher diffusion of haplotype frequencies, complementing existing work for an estimator of selection. We show that, contrary to selection, the estimator has unusual properties because the observed information matrix can explode in finite time whereupon the recombination parameter is learned without error. We also show that the recombination estimator is robust to the presence of selection in the sense that incorporating selection into the model leaves the estimator unchanged. We study the properties of the estimator by simulation and show that its distribution can be quite sensitive to the underlying mutation rates.\\

\end{abstract}


\onehalfspacing

\section{Introduction}
Recombination is a fundamental evolutionary force which shuffles genetic variation along a chromosome and gives rise to new haplotypes not previously seen in a population. It is a major goal of population genetics to infer rates of recombination along the genome and to disentangle its effects from other evolutionary forces such as mutation, selection, migration, and genetic drift. However, the effects of recombination can be difficult to detect; generally the signal of recombination is weak and a single recombination event may leave no discernible trace in a sample of genetic data \citep{hay:etal:biorxiv}. Typically one observes a sample from the state of the population only at the present day, while the evolutionary history of the population, which can be much more informative for recombination, is a latent, unobserved variable. A wide range of inferential methods tackle this problem by positing a generative reproductive model for the population and integrating over all possible evolutionary histories, or by approximating this idea. A popular model is the diffusion limit of a Cannings-type model for recombination, genetic drift, and mutation. Under this limit the evolution of haplotype frequencies follows the \emph{Wright--Fisher diffusion with recombination} \citep{oht:kim:1969:G, oht:kim:1969:GRC} while the genealogical history of a sample is known as the \emph{ancestral recombination graph (ARG)} \citep{gri:mar:1997}. Reconstruction of ARGs is a major current endeavour \citep[see][for recent review]{pen:wol:2020}, and with the very large samples available in recent datasets it becomes ever more necessary to introduce computational and/or model heuristics.

In this paper we address a related question: in the idealised situation in which one observes the entire evolutionary history of a population, as defined via the trajectory of haplotype frequencies in the diffusion limit, can we define an estimator for the recombination rate based on this observation and derive its properties? Although observing the entire sample path of a diffusion is unrealistic in practice, we may regard the corresponding estimator as setting an upper bound on the information about recombination available to us. We note that statistical inference from a continuously observed diffusion is by now a standard problem; see \citet{kut:2004} for textbook treatment for scalar diffusions (though regularity conditions imposed throughout that work preclude most of it applying to the Wright--Fisher diffusion even in one dimension). Further, advances in sequencing technologies are leading to growing availability of genetic data sampled from a population across different times, sometimes over very long timescales, and providing great potential for improved statistical inference \citep{deh:etal:2020}; such datasets can be considered as discrete, noisy versions of the idealised setting studied in this paper.

A motivation for this work is \citet{wat:1979} who derived the maximum likelihood estimator $\hat{s}$ for natural selection from an observation of the trajectory of a Wright--Fisher diffusion (here a diallelic, one-locus model comprising only selection and genetic drift). He found the complete distribution of the estimator. It is worth noting that in this model $\hat{s}$ does not enjoy the usual desirable asymptotic properties such as consistency, since one of the alleles will almost surely go extinct in finite time and thus the total information available about the parameter up to time $T$ remains finite as $T\to\infty$. If we introduce  bidirectional recurrent mutation to the model then it becomes ergodic, and \citet{san:etal:2022} have recently shown that in this situation the estimator enjoys the properties of consistency (uniformly over compact subsets of the parameter space) as well as asymptotic normality and asymptotic efficiency. We will see that, with or without mutation, the estimator for recombination behaves very differently to that of selection because the `information' (defined formally below) can become infinite in finite time. Essentially, in a model of selection the signal-to-noise ratio for the selection parameter remains finite on hitting a boundary of the simplex of possible frequencies, while for the recombination parameter it may not. We will see that if the information becomes infinite then the maximum likelihood estimator (MLE) for recombination becomes exact, $\hat{\rho}_{\text{MLE}} = \rho$.

The paper is structured as follows. In \sref{sec:likelihood} we summarise likelihood theory for a continuously observed diffusion and specialise it to the infinitesimal variance of a Wright--Fisher diffusion. In \sref{sec:MLE-theory} we derive the MLE for a general Wright--Fisher diffusion with arbitrary infinitesimal drift subject only to the constraint that the drift is linear in its unknown parameters. We then specialise this to the model of our primary interest, a multi-locus model with unknown recombination rate. Throughout we focus on the two-locus case which illustrates the main ideas without complicating the notation. Our main results are to derive an expression for the MLE and to show that if the information explodes then it is possible to learn the recombination parameter without error. \sref{sec:selection} studies the impact of the presence of selection on this estimator, and in \sref{sec:simulation} we conduct a simulation study to investigate the empirical properties of the MLE. We discuss some potential directions for future work in \sref{sec:discussion}.

\section{Likelihood in diffusion paths}
\label{sec:likelihood}
\subsection{General case}
\label{sec:general}
We first give a summary of general likelihood-based inference for the parameters of a diffusion before specialising to the Wright--Fisher diffusion. Let $\{X(t):\: t\geq 0\}$ be a $d$-dimensional diffusion process and suppose its path $\{X(t):\:t\in[0,T]\}$ is observed up to time $T$. The generator of the diffusion has a form
\begin{equation}
{\cal L} = \frac{1}{2}\sum_{i,j=1}^d V_{ij}(x)\frac{\partial^2}{\partial x_i\partial x_j} + \sum_{i=1}^d \mu_i(x;\varphi)\frac{\partial}{\partial x_i},
\label{gen:0} 
\end{equation}
where the model has $r$ parameters $\varphi = (\varphi_1,\dots,\varphi_r)^\top$ in a parameter space $\Theta$. We assume the drift $\mu = (\mu_1\dots,\mu_d)^\top$ can be written in the form
\begin{equation}
\label{eq:driftdecomp}
\mu(x; \varphi) = c(x) + a(x; \varphi),
\end{equation}
with $a(x;\varphi_0) \equiv 0$ for a fixed reference parameter $\varphi_0$, and $c(\cdot)$ does not contain any parameters to be estimated. (For example, later we will estimate the rate of recombination in the presence of recurrent mutation with the latter having rates fixed and known. Then $a(\cdot;\varphi)$ will correspond to the contribution of recombination while $c(\cdot)$ will correspond to the contribution of mutation, containing known mutation parameters.)

We will denote the corresponding path measure on continuous functions from $[0,T]$ to $\bbR^d$ by $\bbP^{(T)}_\varphi$.

If the $d \times d$ matrix $V = (V_{ij})$ is non-singular for almost all $t\in[0,T]$ and, for each $\varphi\in\Theta$,
\begin{equation}
\label{eq:information-finite}
\bbP^{(T)}_\varphi(I_{ij} < \infty,\, i,j=1,\dots,r) = 1,
\end{equation}
where $I_T=(I_{ij})$ is the $r\times r$ \emph{observed information matrix}
\begin{equation}
\label{eq:information}
I_T = \int_0^T Z(X(t);\varphi)^\top V^{-1}(X(t))Z(X(t);\varphi) \dd t, \qquad Z_{ij}(x;\varphi) = \frac{\partial a_i(x; \varphi)}{\partial \varphi_j},
\end{equation}
then the likelihood for $\varphi$ takes the form of a Radon--Nikodym derivative
\[
L_T(\varphi) = \frac{\dd \bbP_{\varphi}^{(T)}}{\dd \bbP_{\varphi_0}^{(T)}}
\]
given with respect to a dominating measure which here we have chosen to be the model parametrised by $\varphi_0$ so that $\bbP_{\varphi_0}^{(T)}$ is the distribution over paths with drift $c$. Under these conditions, the likelihood takes the form
\begin{equation}
L_T(\varphi) = \exp\left(\int_0^T a(X(t);\varphi)^\top V(X(t))^{-1}\dd \widetilde{X}(t) - \frac{1}{2}\int_0^Ta(X(t);\varphi)^\top V(X(t))^{-1}a(X(t);\varphi)\dd t \right),
\label{likelihood:0}
\end{equation}
where
\begin{equation}
\label{eq:Xtilde}
\widetilde{X}(t) = X(t) - \int_0^t c(X(s)) \dd s.
\end{equation}
The first integral in \eqref{likelihood:0} is with respect to the path $\{\widetilde{X}(t):\:t\in[0,T]\}$ and the second is with respect to $t$.

A heuristic way to understand equation \eqref{likelihood:0} is as follows. Let $\Delta X(t) = X(t+\Delta t) - X(t)$. The distribution of $\Delta X(t)$ given $X(t)=x$ is taken as approximately normal with mean $\mu(x)\Delta t$ and covariance matrix $V(x)\Delta t$  as $\Delta t \to 0$. If $V(x)$ is non-singular, then the quadratic form in the exponent of the normal density of $\Delta X(t)$ is
\begin{multline*}
\big [\Delta X(t) - \mu(x)\Delta t\big ]^\top \big [\Delta t V(X)\big ]^{-1}\big [\Delta X(t) - \mu(x)\Delta t\big ] =\\
\mu(x)^\top V(x)^{-1}\mu(x)\Delta t - 2\mu(x)^\top V(x)^{-1}\Delta X(t) + \cO((\Delta t)^2).
\label{normal:0}
\end{multline*}
We are expressing this density with respect to another normal density with mean $c(x)\Delta t$ and covariance matrix $V(x)\Delta t$, and thus we subtract the corresponding quadratic form
\[
c(x)^\top V(x)^{-1}c(x)\Delta t - 2c(x)^\top V(x)^{-1}\Delta X(t) + \cO((\Delta t)^2).
\]
After some rearrangement, letting $\Delta t \to 0$, and integrating from $0$ to $T$, we recover the quadratic form appearing in \eqref{likelihood:0}. Note that the likelihood ignores the terms $|V(x)|$ in the diffusion, since we assume that the likelihood is with respect to a parametric form only for $\mu$. (Statistical inference for parameters of $V$ would be trivial in this setting, since $V$ is identifiable from the path via its quadratic variation.) See \citet[Ch.~9]{bas:pra:1980} and \citet[Ch.~6]{klo:etal:2003} for further details on the general case.

\subsection{Wright--Fisher diffusion}
The family of Wright--Fisher diffusions has generator \eqref{gen:0} with diffusion coefficient of the form
\[
V_{ij}(x) = x_i(\delta_{ij}-x_j),
\]
where $\delta_{ij}$ denotes the Kronecker delta (i.e.~$\delta_{ij} = 1$ if $i=j$ and $\delta_{ij} = 0$ if $i\neq j$). The diffusion takes values in the simplex
\[
\Delta_{d-1} := \left\{ x \in [0,1]^d:\: \sum_{i=1}^d x_i = 1\right\},
\]
and the domain of $\cL$ is $\cD(\cL) = C^2(\gD_{d-1})$, twice continuously differentiable functions with domain $\Delta_{d-1}$. For now we continue to leave the drift in the form \eqref{eq:driftdecomp} but otherwise unspecified.

The matrix $V(x)$ is singular since $\sum_{i=1}^d x_i = 1$. Our first task, then, is to modify the results from Section \ref{sec:general} to accommodate this issue. We achieve this by studying the first $d-1$ coordinates of $X$, whose infinitesimal covariance matrix $V^*(x)$ is non-singular. Fortunately, its inverse $V^*(x)^{-1}$ takes on a particularly simple form, as we now show.

\begin{theorem}
\label{thm:likelihood}
Assume \eqref{eq:information-finite} holds for a Wright--Fisher diffusion with drift coefficient $\mu(x;\varphi) = c(x) + a(x; \varphi)$ and diffusion coefficient $V = (V_{ij})$, $V_{ij}(x) = x_i(\gd_{ij}-x_j)$. Then the likelihood is
\begin{equation}
	\label{likelihood}
L_T(\varphi) = \exp\left(\int_{0}^{T}\sum_{i=1}^d\frac{a_i(X(t);\varphi)}{X_i(t)}\dd \widetilde{X}_i(t) - \frac{1}{2}\int_0^T\sum_{i=1}^d\frac{a_i(X(t); \varphi)^2}{X_i(t)}\dd t\right),
\end{equation}
with $\widetilde{X}$ given by \eqref{eq:Xtilde}.
\end{theorem}

\begin{proof}
We consider the diffusion $(X_1(t),\dots,X_{d-1}(t))$ with drift $\mu^*(x;\varphi) = (\mu_1(x;\varphi),\dots,\mu_{d-1}(x;\varphi))^\top$ and non-singular $(d-1)\times(d-1)$ covariance matrix $V^*(x)$. Define $X_d(t) = 1-\sum_{i=1}^{d-1}X_i(t)$ and $\mu_d(x;\varphi) = -\sum_{i=1}^d\mu_i(x;\varphi)$. It follows from standard normal theory, for example \citet[p520--521]{ken:etal:1994}, that
\begin{equation}
\label{eq:Vinverse}
[V^*(x)^{-1}]_{ij} = \big (x_d^{-1} + x_i^{-1}\delta_{ij}\big ).
\end{equation}
We know that $\sum_{i=1}^d a_i(x;\varphi) = 0$ and $\sum_{i=1}^d \dd \widetilde{X}_i(t) = 0$ (since both $\sum_{i=1}^d \dd X_i(t) = 0$ and $\sum_{i=1}^d c_i(X(t)) = 0$, the latter required for $X$ to take values in $\gD_{d-1}$ when $\varphi = \varphi_0$), so
\begin{align}
a^*(X(t); \varphi)^\top V^*(X(t))^{-1} \dd \widetilde{X}(t) &= \sum_{i=1}^{d-1}\sum_{j=1}^{d-1}a_i(X(t); \varphi)(X_d(t)^{-1} + \delta_{ij}X_i(t)^{-1})\dd \widetilde{X}_j(t) \notag\\
&= \frac{a_d(X(t); \varphi)}{X_d(t)}\dd \widetilde{X}_d(t) + \sum_{i=1}^{d-1}\frac{a_i(X(t); \varphi)}{X_i(t)}\dd \widetilde{X}_i(t) \notag\\
&= \sum_{i=1}^d \frac{a_i(X(t); \varphi)}{X_i(t)}\dd \widetilde{X}_i(t), \label{eq:like1stterm}
\end{align}
with a similar calculation for
\begin{align}
a^*(X(t); \varphi)^\top V^*(X(t))^{-1} a^*(X(t); \varphi) &= \sum_{i=1}^{d-1}\sum_{j=1}^{d-1}a_i(X(t); \varphi)a_j(X(t); \varphi){V_{ij}^*(X(t))}^{-1} \notag\\
&= \sum_{i=1}^d \frac{a_i(X(t); \varphi)^2}{X_i(t)}.\label{eq:like2ndterm}
\end{align}
Substituting \eqref{eq:like1stterm} and \eqref{eq:like2ndterm} into \eqref{likelihood:0} yields \eqref{likelihood}.
\end{proof}

\section{Theory for maximum likelihood estimators}
\label{sec:MLE-theory}
\subsection{General Wright--Fisher diffusion}
\label{sec:MLE}
Our next goal is to derive an MLE for the parameters $\varphi$ of a Wright--Fisher diffusion. This is found by differentiating the log-likelihood with respect to the parameters. In all the examples we encounter, the drift is a linear function of the parameters so for the remainder of this article we assume $a(x;\varphi)$ to be of the form
\begin{equation}
\label{eq:linear}
a_i(x;\varphi) = \sum_{k=1}^r Z_{ik}(x)\varphi_k,
\end{equation}
where $Z_{ik}(x) = \frac{\partial a_i(x;\varphi)}{\partial \varphi_k}$ does not depend on $\varphi$. To avoid issues of identifiability we suppose that the columns of $Z = (Z_{ij})$ are linearly independent functions. Then from \thmref{thm:likelihood} the log-likelihood is a quadratic function
\[
\log L_T(\varphi) = \sum_{k=1}^r \varphi_k \int_{0}^{T}\sum_{i=1}^d\frac{Z_{ik}(X(t))}{X_i(t)}\dd \widetilde{X}_i(t) - \frac{1}{2}\sum_{k=1}^r\sum_{l=1}^r \varphi_k\varphi_l\int_0^T\sum_{i=1}^d\frac{Z_{ik}(X(t))Z_{il}(X(t))}{X_i(t)}\dd t,
\]
with a unique maximum, $\hat{\varphi}$, in $\bbR^r$, which is the solution of the set of equations for $k=1,\dots ,r$:
\begin{equation}
0 =\int_{0}^{T}\sum_{i=1}^d\frac{Z_{ik}(X(t))}{X_i(t)}\dd \widetilde{X}_i(t) -\sum_{l=1}^r \varphi_l \int_{0}^{T}\sum_{i=1}^d\frac{Z_{ik}(X(t))Z_{il}(X(t))}{X_i(t)}\dd t.
\label{general_mle}
\end{equation}
The equations \eqref{general_mle} are familiar in regression theory. Now denote the ($r\times 1$) vector $Y = \big (Y_{k}\big )$ with elements
\[
Y_k = \int_{0}^{T}\sum_{i=1}^d\frac{Z_{ik}(X(t))}{X_i(t)}\dd\widetilde{X}_i(t),
\]
and let $\Sigma (X(t)) = \text{diag}(X_i(t))$. Then the equations \eqref{general_mle} can be written
\[
\Biggl [\int_0^T Z(X(t))^\top\Sigma^{-1}(X(t)) Z(X(t)) \dd t\Biggr ]\hat{\varphi} = Y.
\]
Continuing to assume \eqref{eq:information-finite}, the matrix on the left-hand side of the previous equation is non-singular (see \citet[p223--224]{bas:pra:1980}, \citet[p231]{klo:etal:2003}) and hence we arrive at the form
\begin{equation}
\label{eq:varphi-hat}
\hat{\varphi} = \Biggl [\int_0^T Z(X(t))^\top\Sigma^{-1}(X(t)) Z(X(t)) \dd t\Biggr ]^{-1}Y.
\end{equation}

Of course if $\Theta \subset \bbR^k$ then it is not guaranteed that $\hat{\varphi} \in \Theta$, and $\hat{\varphi}$ must be adjusted appropriately to ensure it is the MLE. An example of this adjustment is given later.

The observed information matrix \eqref{eq:information} is a key quantity in telling us about how informative the data is for $\varphi$. For this model, the observed information matrix $I_T$ has elements
\begin{equation}
 I_{kl} = \int_{0}^{T}\sum_{i=1}^d
\frac{1}{X_i(t)}Z_{ik}(X(t))Z_{il}(X(t))
\dd t
\label{obsinf}
\end{equation}
using linearity of $a_i(x;\varphi)$; the information matrix does not depend on $\hat{\varphi}$. The expression \eqref{eq:varphi-hat} for $\hat{\varphi}$ can be written
\[
\hat{\varphi} = I_T^{-1}Y.
\]

\subsubsection{Deterministic model}
As a check on the expression for $\hat{\varphi}$, we can ask for the estimator we would obtain if the observed trajectory is that of the deterministic model
\begin{equation}
\frac{\dd x_i}{\dd t} = \mu_i(x;\varphi), \qquad i=1,\dots,d.
\label{deter:0}
\end{equation}
Now from \eqref{eq:driftdecomp} and \eqref{eq:Xtilde} we find
\[
\dd \widetilde{x} = \dd x - c(x) \dd t = (c(x) + a(x; \varphi))\dd t - c(x) \dd t = a(x; \varphi) \dd t,
\]
and so we can substitute this expression for $\dd \widetilde{x}$ into the likelihood equation \eqref{general_mle} to obtain that $\hat{\varphi}$ is a solution to
\begin{equation}
0 =\int_{0}^{T}\sum_{i=1}^d\frac{[a_i(x(t);{\varphi}) - a_i(x(t);{\hat{\varphi}})]}{x_i(t)}\frac{\partial a_i(x(t);\hat{\varphi})}{\partial \varphi_l}\dd t.\label{general_mle:1}
\end{equation}
Owing to the factor $a_i(x(t);{\varphi}) - a_i(x(t);{\hat{\varphi}})$, it is clear that a solution to the likelihood equation is given by $\hat{\varphi}=\varphi$. It is reassuring that the estimator is well behaved even in this crude level of approximation; the trajectory defined by \eqref{deter:0} is not a realisation from the assumed model since it is a path of bounded variation.

\subsection{Neutral two-locus model}
We now turn to our main result, an expression for the MLE for the recombination parameter $\rho\in [0,\infty) =: \Theta$. Consider a neutral two-locus model in which there are $K$ possible alleles at the first locus, locus A, and $L$ possible alleles at the second, locus B. The haplotype of an individual is denoted $(i,j) \in \{1,\dots,K\} \times \{1,\dots,L\}$, and its frequency in the population is $x_{ij}$. Note that to reconcile this double-index notation with previous sections we must implicitly stack the $KL$ possible haplotypes in some agreed order into a vector of length $d = KL$. We will switch between the two notations as required. To emphasise when haplotypes have been stacked we will use a bold index, so $x_{\bfi}$ denotes the frequency of haplotype $\bfi$, $\bfi = 1,\dots,d$.

The model is completed by specifying the drift. Here it is of the form
\begin{align*}
a_{ij}(x;\rho) &= \rho (x_{i\cdot}x_{\cdot j} - x_{ij}), \qquad i=1,\dots,K;\; j=1,\dots,L,
\end{align*}
where $x_{i\cdot} := \sum_{l=1}^L x_{il}$ and $x_{\cdot j} := \sum_{k=1}^K x_{kj}$. Recombination occurs between the two loci at rate $\rho$; specifically this is a model of the homologous crossing-over that takes place during meiosis. To simplify later results, we omit the conventional factor of $1/2$ in the recombination rate parameter.

In much of what follows the choice for $c(x)$ is immaterial, but for concreteness we will set
\begin{align*}
c_{ij}(x) &= \frac{\tA}{2}\sum_{k=1}^K x_{kj}(P_{ki}^\A - \delta_{ik}) + \frac{\tB}{2}\sum_{l=1}^L x_{il}(P_{lj}^\B - \delta_{jl}).
\end{align*}
Here mutation takes place at locus A and B at respective rates $\tA/2$ and $\tB/2$ on the timescale of the diffusion. When a mutation occurs, the change in allele is governed by the $K\times K$ and $L\times L$ mutation transition matrices $P^\A$ and $P^\B$ (i.e.~if a mutation occurs at locus A on haplotype $(k,j)$ then it mutates to haplotype $(i,j)$ with probability $P^\A_{ki}$, $i=1,\dots,K$; similarly for $P^\B$). We allow $\tA,\tB \geq 0$, so the model may or may not be ergodic.

Note the separate roles for the two components of the drift: here it is only $\rho$ to be estimated, with the other parameters appearing in $c(\cdot)$ considered known. The likelihood is expressed with respect to the parametrisation $\rho_0 = 0$, a model in which the two loci are completely linked but the mutation parameters are the same.

Using the results of \sref{sec:MLE}, for this model the log-likelihood is
\begin{align}
\log L_T(\rho) = {}& \rho\int_{0}^{T}\sum_{i=1}^K\sum_{j=1}^L\Biggl (\frac{X_{i\cdot}(t)X_{\cdot j}(t)}{X_{ij}(t)} - 1\Biggr )\dd\widetilde{X}_{ij}(t)\notag\\
& {}- \frac{1}{2}\rho^2\int_0^T\sum_{i=1}^K\sum_{j=1}^L\frac{(X_{ij}(t)-X_{i\cdot}(t)X_{\cdot j}(t))^2}{X_{ij}(t)}\dd t \nonumber\\
= {}&\rho\int_{0}^{T}\sum_{i=1}^K\sum_{j=1}^L\frac{X_{i\cdot}(t)X_{\cdot j}(t)}{X_{ij}(t)}\dd\widetilde{X}_{ij}(t) 
- \frac{1}{2}\rho^2\int_0^T\sum_{i=1}^K\sum_{j=1}^L\frac{(X_{ij}(t)-X_{i\cdot}(t)X_{\cdot j}(t))^2}{X_{ij}(t)}\dd t,
\label{twologlike}
\end{align}
where for the second equality we recall $\sum_{i=1}^K \sum_{j=1}^L \dd \widetilde{X}_{ij}(t) = 0$, with
\[
\widetilde{X}_{ij}(t) = X_{ij}(t) - \int_0^t c_{ij}(X(s))\dd s, \qquad i=1,\dots,K;\: j=1,\dots,L.
\]
The estimator $\hat{\rho}$ is therefore
\begin{equation}
\hat{\rho} =
\frac{
\ds\int_{0}^{T}\sum_{i=1}^K\sum_{j=1}^L\frac{X_{i\cdot}(t)X_{\cdot j}(t)}{X_{ij}(t)}\dd\widetilde{X}_{ij}(t) }
{\ds\int_0^T\sum_{i=1}^K\sum_{j=1}^L\frac{(X_{ij}(t)-X_{i\cdot}(t)X_{\cdot j}(t))^2}{X_{ij}(t)}\dd t},
\label{estimate}
\end{equation}
and the observed information is
\begin{equation}
\label{eq:information-recombination}
I_T = \int_0^T\sum_{i=1}^K\sum_{j=1}^L\frac{(X_{ij}(t)-X_{i\cdot}(t)X_{\cdot j}(t))^2}{X_{ij}(t)}\dd t.
\end{equation}
The denominator in $\hat{\rho}$ and the information can be simplified to
\[
I_T = \int_0^T\sum_{i=1}^K\sum_{j=1}^L \frac{X_{i\cdot}(t)^2X_{\cdot j}(t)^2}{X_{ij}(t)}\dd t -T.
\]
It is worth remarking on the functional form \eqref{twologlike} for $\log L_T(\rho)$. This is a polynomial in $\rho$ and we can think of a trade-off between the order of the polynomial and the complexity of its coefficients. In this model we have a particularly simple quadratic polynomial in $\rho$, order only two, with the benefit of knowing that the function is convex with a unique finite maximum (since the coefficient of $\rho^2$ is negative). The price we pay is that the coefficients of the polynomial are highly cumbersome in the sense that they are given as integrals over the sample path of a diffusion. Contrast this with the dual coalescent model in which the likelihood for an observed sample path of an ARG would be a product of exponential waiting time densities times a product of rational functions for the transitions of the jump chain. With many possible jumps, these rational functions may be constructed from polynomials in $\rho$ of very high order, though their coefficients are much simpler than the stochastic integrals encountered here. In a coalescent model, the shape of the likelihood curve as a function of $\rho$ can be rather complicated, even exhibiting local minima when integrating over ARGs \citep{jen:son:2009:G}.
\subsubsection{Deterministic model}
Is $\hat{\rho}$ in \eqref{estimate} a reasonable estimate? Again we can check what happens when $X$ solves a deterministic model. Setting $\tA=\tB=0$ for the moment, the deterministic model is
\begin{equation}
\frac{\dd x_{ij}}{\dd t} = \rho (x_{i\cdot}x_{\cdot j}-x_{ij}), \qquad \qquad i=1,\dots,K;\: j=1,\dots,L.
\label{de}
\end{equation}
and substituting $\dd x_{ij}$ directly into \eqref{estimate} again shows that $\hat{\rho} = \rho$.

We can further describe the evolution of $I_T$. Note that in this deterministic model (summing \eqref{de} over $j$):
\[
\frac{\dd x_{i\cdot}}{\dd t} = 0, \qquad i=1,\dots,K.
\]
Therefore $x_{i\cdot}(t) = x_{i\cdot}(0)$ for all $t \geq 0$, and similarly for $x_{\cdot j}(t)$. 
The solution to (\ref{de}) is then
\begin{equation}
x_{ij}(t) = x_{ij}(0)e^{-\rho t} + x_{i\cdot}(0)x_{\cdot j}(0)(1 - e^{-\rho t}).
\label{desoln}
\end{equation} 
We have that
\begin{equation*}
\log \left(\frac{x_{ij}(T)}{x_{ij}(0)}\right) = \int_0^T\frac{\dd x_{ij}}{x_{ij}} = \rho\left(\int_0^T \frac{x_{i\cdot}x_{\cdot j}}{x_{ij}}\dd t - T\right),
\end{equation*}
so (provided $\rho > 0$):
\begin{equation}
I_T = \rho^{-1}\sum_{i=1}^K\sum_{j=1}^Lx_{i\cdot}(0)x_{\cdot j}(0)\log\left(\frac{x_{ij}(T)}{x_{ij}(0)}\right).
\label{Info}
\end{equation}
The limit information is therefore
\[
\lim_{T\to\infty} I_T = \rho^{-1}\sum_{i=1}^K\sum_{j=1}^Lx_{i\cdot}(0)x_{\cdot j}(0)\log\left(\frac{x_{i\cdot}(0)x_{\cdot j}(0)}{x_{ij}(0)}\right).
\]
As far as the deterministic model goes, the information is in the transient phase until the frequencies come to equilibrium. The accumulated information $I_T$ remains finite as $T\to\infty$. In a stochastic model on the other hand, we will see that the injection of noise allows $I_T \to\infty$ as $T\to\infty$. We note that one should regard this contrasting behaviour with caution: it does \emph{not} mean that the estimator is consistent only in the stochastic setting. We have just seen that $\hat{\rho} = \rho$ in the deterministic setting, which is trivially consistent, and creates a paradox when we try to reconcile this fact with the asymptotic finiteness of $I_T$. The paradox is resolved by noting that the data-generating mechanism differs from the one assumed in designing the estimator. Had we assumed a deterministic model throughout our analysis then, since the parameter is simply a rate appearing in an observed ODE, the `likelihood' would be a point mass on the true rate and the MLE would be equal to that true rate. The `information' in this setting, being the curvature of the log-likelihood, is immediately infinite. The fact that $\hat{\rho} = \rho$ demonstrates that the estimator adapts automatically to a change in data-generating mechanism. The quantity $I_T$ could be regarded not as the information under the true model but as a way of quantifying `the informativeness of the deterministic trajectory under stochastic assumptions'. It is this quantity that remains finite as $T\to\infty$.

It is possible to repeat these calculations for a model with $\tA, \tB > 0$; that is, to solve the deterministic mutation-recombination equation. Again $I_T$ converges to a finite limit; see \aref{sec:appendix}.

\subsubsection{Stochastic differential equation interpretation}
We can find an expression for the error associated with $\hat{\rho}$ by regarding $X(t)$ as the solution to a stochastic differential equation (SDE):
\begin{equation}
\label{eq:SDE}
\dd X(t) = [c(X(t)) + a(X(t); \varphi)] \dd t + \gs(X(t)) \dd W(t), \qquad X(0) = x(0),
\end{equation}
where $W$ is a $(d-1)$-dimensional Brownian motion and $\gs(x)$ is a (non-unique) $d\times (d-1)$ matrix satisfying $\gs(x)\gs(x)^\top = V(x)$.

There are various ways to define $\gs$ subject to this constraint. It is common to ask for $\gs$ to be lower triangular by applying a Cholesky decomposition to $V$. In the case of the covariance matrix of the Wright--Fisher diffusion, the Cholesky decomposition is given analytically by \citet{sat:1976}, though it is one that explodes at the boundaries.

\begin{proposition}
\label{prop:error}
The error associated with $\hat{\rho}$ is
\begin{equation}
\hat{\rho} - \rho = -\frac{\displaystyle\int_{0}^{T} \sum_{\bfi=1}^{d} \frac{D_{\bfi}(t)}{X_{\bfi}(t)} \sum_{\bfj=1}^{d-1}\gs_{\bfi\bfj}(X(t)) \dd W_{\bfj}(t)}{\displaystyle\int_0^T \sum_{\bfi=1}^d \frac{D_{\bfi}(t)^2}{X_{\bfi}(t)}\dd t}, \label{eq:rho-error}
\end{equation}
where $D_{\bfi}(t) = X_{i_1i_2}(t) - X_{i_1\cdot}(t)X_{\cdot i_2}(t)$ is the coefficient of linkage disequilibrium for haplotype $\bfi = (i_1,i_2)$.
\end{proposition}
\begin{proof}
Rearranging \eqref{estimate} slightly and using $\sum_{\bfi=1}^d \dd \widetilde{X}_{\bfi}(t) = 0$, we have
\[
\hat{\rho}I_T =-\int_{0}^{T}\sum_{\bfi=1}^d\frac{D_{\bfi}(t)}{X_{\bfi}(t)}\dd \widetilde{X}_{\bfi}(t).
\]
Now substituting for $\dd\widetilde{X}_{\bfi}(t) = \dd X_{\bfi}(t) - c_{\bfi}(X(t))\dd t$ using \eqref{eq:SDE},
\[
\hat{\rho}I_T = \rho I_T - \int_{0}^{T} \sum_{\bfi=1}^{d} \frac{D_{\bfi}(t)}{X_{\bfi}(t)} \sum_{\bfj=1}^{d-1}\gs_{\bfi\bfj}(X(t)) \dd W_{\bfj}(t),
\]
which leads to \eqref{eq:rho-error}.
\end{proof}
Thus the bias and mean squared error of $\hat{\rho}$ are given respectively by the expectation of the term on the right-hand side of \eqref{eq:rho-error} and the expectation of its square. Estimators of this form are not unbiased in general \citep[p218]{bas:pra:1980}.

\subsubsection{Corrected MLE}
\label{sec:corrected}
There are two problems with the estimator $\hat{\rho}$ defined in \eqref{estimate}. First, the parameter space is $\Theta = [0,\infty)$ but we cannot ensure $\hat{\rho} \geq 0$. (Although $\rho < 0$ is biologically unrealistic, mathematically it is nonetheless a valid model and a sample path may point to this region of the parameter space if $\dd \widetilde{X}_{ij}(t)$ is sufficiently negative.) This is easily corrected by applying a rectified linear unit, $\max\{0,\hat{\rho}\}$. The second issue is more serious: it is not guaranteed that \eqref{eq:information-finite} holds. In other words, we have not ruled out the possibility that $I_T$ explodes in finite time. For observations for which $I_T < \infty$, we can still interpret \eqref{twologlike} as a quasi-log-likelihood function \citep[p231]{klo:etal:2003}, but otherwise we must treat $L_T(\rho)$ as a \emph{generalized} density valid only until the stopping time
\begin{equation}
\label{eq:stopping}
S := \inf\left\{t\in [0,\infty):\: I_t = \infty\right\}.
\end{equation}
See \citet[Ch.~6]{lip:shi:2001:I} and \citet{mij:etal:2012} for further discussion on this subtle point. Writing $\hat{\rho} = \hat{\rho}_T$ for the estimator in \eqref{estimate}, we define the following corrected estimator:
\begin{equation}
\label{eq:rho-mle}
\hat{\rho}_{\text{MLE}} := \bbI_{[0,S)}(T)\max\left\{0,\hat{\rho}_T\right\} + \bbI_{[S,\infty)}(T)\lim_{t\uparrow S}\hat{\rho}_t.
\end{equation}
A similar issue arises in the estimation of the immigration rate of the continuous branching with immigration (CBI) diffusion, where a related correction is proposed \citep[in particular Theorem 2(iv)]{ove:1998}. The subscript in \eqref{eq:rho-mle} rather suggestively posits this quantity as \emph{the} MLE; this is proven shortly, in Corollary \ref{cor:MLE}. Although taking $\lim_{t\uparrow S}\hat{\rho}_t$ in \eqref{eq:rho-mle} might seem to be unstable, that this is the appropriate correction to our estimator is justified by the following theorem.
\begin{theorem}
\label{thm:error}
If $S \leq T$ then $\hat{\rho}_{\text{MLE}} = \rho$ with probability 1.
\end{theorem}
\begin{proof}
From the definition \eqref{eq:rho-mle} of $\hat{\rho}_{\text{MLE}}$ it suffices to show that $\hat{\rho} \to \rho$ as $T \uparrow S$. Let
\[
N_T = -\int_{0}^{T\wedge S} \sum_{\bfi=1}^{d} \frac{D_{\bfi}(t)}{X_{\bfi}(t)} \sum_{\bfj=1}^{d-1}\gs_{\bfi\bfj}(X(t)) \dd W_{\bfj}(t).
\]
This is a continuous, stopped martingale with $N_0 = 0$ and quadratic variation
\begin{align*}
\langle N \rangle_T &=  \left\langle-\int_{0}^{T\wedge S} \sum_{\bfi=1}^{d} \frac{D_{\bfi}(t)}{X_{\bfi}(t)} \sum_{\bfj=1}^{d-1}\gs_{\bfi\bfj}(X(t)) \dd W_{\bfj}(t)\right\rangle_T\\
&= \int_{0}^{T\wedge S}\sum_{\bfi=1}^d\sum_{\bfj=1}^{d-1} \frac{D_{\bfi}(t)}{X_{\bfi}(t)}\gs_{\bfi\bfj}(X(t)) \sum_{\bfk=1}^d  \frac{D_{\bfk}(t)}{X_{\bfk}(t)}\gs_{\bfk\bfj}(X(t))\dd t\\
&= \int_{0}^{T\wedge S}\sum_{\bfi=1}^d\sum_{\bfk=1}^{d} \frac{D_\bfi(t)}{X_{\bfi}(t)}\frac{D_{\bfk}(t)}{X_{\bfk}(t)} V_{\bfi\bfk}(X(t)) \dd t\\
&= \int_{0}^{T\wedge S}\sum_{\bfi=1}^d \frac{D_{\bfi}(t)^2}{X_{\bfi}(t)}\dd t\\
&= I_{T\wedge S}.
\end{align*}
where the second equality uses that $\langle\cdot,\cdot\rangle$ is a bilinear form and $\langle \dd W_{\bfj}, \dd W_{\bfl}\rangle = \gd_{\bfj\bfl}\dd t$. Thus by the law of large numbers for local martingales \citep[Ch.~V.1, Exercise 1.16, p186]{rev:yor:1999},
\[
\lim_{T\to\infty} \frac{N_T}{I_{T\wedge S}} = 0 \qquad \text{with probability 1 on }\{I_\infty = \infty\}.
\]
The limit as $T\uparrow S$ is the same. But $N_T/I_{T\wedge S}$ is precisely the error $\hat{\rho}-\rho$ given in \propref{prop:error}, so $\hat{\rho}\to\rho$ as $T\uparrow S$ with probability 1.
\end{proof}
\begin{corollary}
\label{cor:MLE}
$\hat{\rho}_{\text{MLE}}$ is the MLE for $\rho$.
\end{corollary}
\begin{proof}
This follows since we have separately verified that it is the MLE on $\{T < S\}$ and on $\{S \leq T\}$. In the latter case $\rho$ is identifiable, so $L_T(\rho)$ is zero anywhere other than the true value.
\end{proof}
\begin{corollary}
The error associated with $\hat{\rho}_{\text{MLE}}$ is
\[
\hat{\rho}_{\text{MLE}} - \rho = -\bbI_{[0,S)}(T) \times \begin{cases}
\frac{\displaystyle\int_{0}^{T} \sum_{\bfi=1}^{d} \frac{D_{\bfi}(t)}{X_{\bfi}(t)} \sum_{\bfj=1}^{d-1}\gs_{\bfi\bfj}(X(t)) \dd W_{\bfj}(t)}{\displaystyle\int_0^T \sum_{\bfi=1}^d \frac{D_{\bfi}(t)^2}{X_{\bfi}(t)}\dd t}, & \hat{\rho} \geq 0,\\
\rho, & \hat{\rho} < 0,
\end{cases}
\]
where we recall that $\hat{\rho}$ is the uncorrected estimator given in \eqref{estimate}.
\end{corollary}
\begin{proof}
This follows by combining \thmref{thm:error} and \propref{prop:error}.
\end{proof}
The relevance of \thmref{thm:error} is: If the sample path is such that $I_T = \infty$, then we learn $\rho$ without error. Inspecting the form of $I_T$ in \eqref{eq:information-recombination}, we see that its integrand is locally integrable in the interior of $\gD_{d-1}$. Thus for $I_T = \infty$ it is necessary to have at least one haplotype frequency $X_{ij}(t) \to 0$ before time $T$. We should expect the same phenomenon when inferring the \emph{mutation} parameters in a one-locus model, where hitting one of the boundaries is completely informative for one of the mutation parameters.

The next result shows that having $I_T = \infty$ is not a hypothetical concern, and the proof makes it clear that explosion of $I_T$ is intimately related with hitting a boundary of $\gD_{d-1}$.

\begin{theorem}
\label{thm:explode}
Suppose that $\rho+\frac{\tA}{2}+\frac{\tB}{2} < \frac{1}{2}$, that mutation is parent-independent (i.e.~$P^\A$ and $P^\B$ each have identical rows), and that $x(0)$ lies in the interior of $\gD_{d-1}$. Then $\bbP(I_T = \infty) > 0$.
\end{theorem}
\begin{proof}
It is clear from the form of $I_T$ in \eqref{eq:information-recombination} that $\{I_T = \infty\}$ will occur if for some $i,j$,
\begin{enumerate}[(i)]
\item For some $\gd > 0$ and for all $t \in [0,T]$, $X(t)$ lies in $A_1 := \{x \in \gD_{d-1}:\: (x_{ij} - x_{i\cdot}x_{\cdot j})^2 > \gd\}$;
\item $T_{\varepsilon}(X_{ij}) := \inf\{t\in[0,\infty):\: X_{ij}(t) = \varepsilon\}$, the first hitting time of $\varepsilon$ by $X_{ij}$, satisfies $T_{0}(X_{ij}) \in (0,T]$;  and
\item The integral
\[
\int_0^{T_{\varepsilon}(X_{ij})} \frac{1}{X_{ij}(t)} \dd t
\]
diverges as $\varepsilon \to 0$. 
\end{enumerate}
Condition (iii) extracts the explosion of $I_T$ from a denominator of its integrand, while condition (i) controls the corresponding numerator. Condition (ii) ensures that such explosion takes place before time $T$.

To study the finiteness or otherwise of the integral in (iii), choose a decomposition $\sigma(x)\sigma(x)^\top = V(x)$ so that the component of the SDE \eqref{eq:SDE} corresponding to $X_{ij}(t)$ has the form
\[
\dd X_{ij}(t) = \mu_{ij}(X_{ij}(t))\dd t + \sqrt{X_{ij}(t)(1-X_{ij}(t))} \dd W(t), \qquad X_{ij}(0) = x_{ij}(0),
\]
for a scalar Brownian motion $W$, where
\[
 \mu_{ij}(X_{ij}(t)) = \rho [X_{i\cdot}(t)X_{\cdot j}(t) - X_{ij}(t)] + \frac{\tA}{2}[X_{\cdot j}(t)P_i^\A - X_{ij}(t)] + \frac{\tB}{2}[X_{i\cdot}(t)P_j^\B - X_{ij}(t)].
\]
The idea is to show that this SDE behaves locally like a one-locus model of \emph{mutation only}. More precisely we will compare $X_{ij}$ to another diffusion which solves the SDE
\[
\dd Z(t) = \frac{\vartheta}{2}[P - Z(t)]\dd t + \sqrt{Z(t)(1-Z(t))} \dd W(t), \qquad Z(0) = x_{ij}(0),
\]
for some $\vartheta \in [0,1)$, $P \in (0,1)$. Choose $\vartheta$ so that $\rho + \frac{\tA}{2} + \frac{\tB}{2} < \frac{\vartheta}{2}$ and choose $P$ so that $x_{i\cdot}(0)x_{\cdot j}(0) < P$, $P^\A_i < P$, and $P^\B_j < P$. Then on the set $A_2 := \{ x \in \gD_{d-1}:\: x_{i\cdot}x_{\cdot j} < P\}$ it is straightforward to verify we have
\[
\mu_{ij}(x) < \frac{\vartheta}{2}(P - x),
\]
and thus by a standard comparison theorem \citep[see Theorem 1.1 and Remark 1.1 in][]{ike:wat:1977}
we can construct a probability space on which $Z(t) \geq X_{ij}(t)$ for all $t\in[0,T_{A_2^\complement})$, where $T_A := \inf\{t\in[0,\infty):\: X(t) \in A\}$. (For the comparison theorem to hold there is a required growth condition on the diffusion coefficient. That this holds follows from the fact that $\sqrt{x(1-x)}$ is $1/2$-H\"older continuous; see also Remark 3.9 on p298 of \citet{eth:kur:1986}.) Thus condition (iii) is implied by the a.s.~divergence of
\[
\int_0^{T_{\varepsilon}(Z)} \frac{1}{Z(t)} \dd t
\]
as $\varepsilon \to 0$, which in turn follows from Lemma 4.4 of \citet{bar:etal:2004}, noting that $\vartheta < 1$ guarantees the 0-boundary for $Z$ is accessible. [Some errors in the proof of Lemma 4.4 are corrected by \citet{tay:2007}.] Tracing our steps backwards, we have shown that condition (iii) holds provided $0< T_0(X_{ij}) \leq T < T_{A_2^\complement}$. Since
\[
\bbP(T_{A_1^\complement} > T,\, 0< T_0(X_{ij}) \leq T < T_{A_2^\complement}) > 0,
\]
we conclude $\bbP(I_T = \infty) > 0$.
\end{proof}
The conditions given in \thmref{thm:explode} simplify our proof, but it seems feasible to substantially weaken them.
\subsection{Testing for the presence of recombination}
It is possible to use $\hat{\rho}_{\text{MLE}}$ to design a likelihood ratio test for the null hypothesis that $\rho_0 = 0$. Using \eqref{twologlike}, the appropriate likelihood ratio statistic is, for $I_T < \infty$,
\[
\gL := 2\log \frac{\dd\bbP^{(T)}_{\hat{\rho}_{\text{MLE}}}}{\dd\bbP^{(T)}_{\rho_0}} = \hat{\rho}_{\text{MLE}}^2I_T, \qquad I_T < \infty.
\]
Under standard assumptions, noting that $\rho_0 = 0$ lies on the boundary of $\Theta$, this has an asymptotic distribution which is an equal mixture between a $\chi^2_1$ distribution and a $\chi^2_0$ distribution under the null hypothesis \citep{sel:lia:1987}. Denote the CDF of this distribution by $F_m$. In particular, to construct a level 5\% test one should reject $\rho_0 = 0$ if $\gL$ exceeds the 95th percentile of $F_m$; equivalently if it exceeds the 90th percentile of a $\chi^2_1$ distribution.

To account for the possibility that $I_T = \infty$ we set
\[
\gL := \begin{cases}
+\infty, & \hat{\rho}_{\text{MLE}} > 0,\\
0, & \hat{\rho}_{\text{MLE}} = 0.
\end{cases}, \qquad I_T = \infty.
\]
The asymptotic null distribution for $\gL$ is now less clear, though we note that continuing to assume $F_m$ would be conservative. We study the power of this test empirically in \sref{sec:simulation}.

\subsection{Multiple loci}
It is possible to extend the above results to a general multi-locus model. The extension is straightforward and we omit many of the lengthy but straightforward calculations.

In a multi-locus model of $\ell$ loci with $K_j$ possible alleles at locus $j$, haplotypes are of the form $\bfi = (i_1,\dots,i_\ell) \in \prod_{j=1}^\ell \{1,\dots,K_j\} =: E$ in a diffusion on $\gD_{d-1}$ with $d = \prod_{j=1}^\ell K_j$ coordinates. Stacking the haplotypes, the diffusion coefficient has entries $V_{\bfi\bfk}(x) = x_{\bfi}(\gd_{\bfi\bfk} - x_{\bfk})$ as usual, and the unknown component of the drift is
\[
a_\bfi(x;\rho_1,\dots,\rho_{\ell-1}) = \sum_{j=1}^{\ell-1} \rho_j(x_{\bfi_{\leq j}}x_{\bfi_{>j}} - x_\bfi),
\]
where $\rho_j$ is the recombination rate between locus $j$ and $j+1$, with each $\rho_j$ to be estimated; $x_\bfi$ is the frequency of haplotype $\bfi$; and we marginalize over a contiguous subset of loci by writing
\begin{align*}
	x_{\bfi_{\leq j}} &= \sum_{i_{j+1}=1}^{K_{j+1}} \cdots \sum_{{i_\ell}=1}^{K_{\ell}} x_{(i_1,\dots,i_\ell)}, &
	x_{\bfi_{> j}} &= \sum_{{i_1}=1}^{K_1} \cdots \sum_{i_j=1}^{K_j} x_{(i_1,\dots,i_\ell)}.
\end{align*}
From \eqref{eq:varphi-hat} and \eqref{obsinf} the joint estimator for $(\rho_1,\dots,\rho_{\ell-1})$ is $\hat{\varrho} = I_T^{-1}Y$ where $I_T$ is $(\ell-1)\times(\ell-1)$ and $Y$ is $(\ell-1)\times 1$ with elements
\begin{align*}
	I_{jk} &= \int_0^T \sum_{\bfi\in E} \frac{\left(X_{\bfi_{\leq j}}(t)X_{\bfi_{> j}}(t) - X_\bfi(t)\right)\left(X_{\bfi_{\leq k}}(t)X_{\bfi_{> k}}(t) - X_\bfi(t)\right)}{X_\bfi(t)}\dd t\\
	&= \int_0^T \sum_{\bfi\in E} \frac{X_{\bfi_{\leq j}}(t)X_{\bfi_{> j}}(t)X_{\bfi_{\leq k}}(t)X_{\bfi_{> k}}(t)}{X_\bfi(t)}\dd t - T,\\
	Y_j &= \int_{0}^{T}\sum_{\bfi\in E} \frac{X_{\bfi_{\leq j}}(t)X_{\bfi_{>j}}(t) - X_\bfi(t)}{X_\bfi(t)} \dd\widetilde{X}_\bfi(t) = \int_{0}^{T}\sum_{\bfi\in E} \frac{X_{\bfi_{\leq j}}(t)X_{\bfi_{>j}}(t)}{X_\bfi(t)} \dd\widetilde{X}_\bfi(t).
\end{align*}

An alternative model is to set $\rho_j = \rho$ for each $j=1,\dots,\ell$ and to construct a single scalar estimator. Then the estimator is
\[
\hat{\varrho} = \frac{\displaystyle\sum_{j=1}^{\ell-1}\int_0^T\sum_{\bfi\in E} \frac{X_{\bfi_{\leq j}}(t)X_{\bfi_{>j}}(t)}{X_\bfi(t)} \dd\widetilde{X}_\bfi(t)}{\displaystyle\sum_{j,k=1}^{\ell-1}\int_0^T \sum_{\bfi\in E} \frac{X_{\bfi_{\leq j}}(t)X_{\bfi_{> j}}(t)X_{\bfi_{\leq k}}(t)X_{\bfi_{> k}}(t)}{X_\bfi(t)}\dd t - (\ell-1)^2 T}.
\]
These estimators should be corrected as in \eqref{eq:rho-mle}.

\section{The effects of natural selection}
\label{sec:selection}
Methods for inference of recombination can be confounded by natural selection \citep{ree:tis:2006, ore:etal:2008, pen:wol:2020}. In this section we investigate the effect of selection on $\hat{\rho}_{\text{MLE}}$, for simplicity returning to a two-locus model, though it should be straightforward to extend these results to general multi-locus models.
\subsection{Confounding by selection}
\label{sec:confounding-sel}
First consider the following: Suppose that, unknown to the investigator, the two loci are under selection---possibly a complicated type with epistatic interaction. What effect does this have on our estimator for $\rho$? More precisely, consider a model in which the component of the drift with parameters to be estimated is still $a_{ij}(x;\rho) = \rho (x_{i\cdot}x_{\cdot j} - x_{ij})$, but the `known' component of the drift is now
\begin{align*}
c_{ij}(x) = {} & \frac{\tA}{2}\sum_{k=1}^K x_{kj}(P_{ki}^\A - \delta_{ik}) + \frac{\tB}{2}\sum_{l=1}^L x_{il}(P_{lj}^\B - \delta_{jl})\\
& {}+ \frac{x_{ij}}{2}\left[\sum_{k=1}^K\sum_{l=1}^L\left(s_{ij,kl}x_{kl} - \sum_{m=1}^K\sum_{n=1}^Ls_{kl,mn}x_{kl}x_{mn}\right)\right], \qquad i=1,\dots,K;\; j=1,\dots,L.
\end{align*}
This is a very general diploid, epistatic model of selection in which the selective advantage of an individual carrying haplotypes $(i,j)$ and $(k,l)$, relative to other individuals, is parametrised by $s_{ij,kl}$. We are interested in the role of selection as a confounder, whereby inference is carried out using the incorrect selection parameters in the dominating measure.

From equation \eqref{estimate}, $c(\cdot)$ has an effect on $\hat{\rho}$ only through the term $\dd \widetilde{X}_{ij}(t) = \dd X_{ij}(t) - c_{ij}(X(t))\dd t$. Therefore, in a model with selection we should adjust \eqref{estimate} by defining a new estimator
\begin{align}
\hat{\rho}_{\text{sel}} = {} &\hat{\rho} - \frac{1}{I_T}\int_0^T \sum_{i=1}^K\sum_{j=1}^L\frac{X_{i\cdot}(t)X_{\cdot j}(t)}{X_{ij}(t)}\notag\\
& \phantom{\hat{\rho} - \frac{1}{I_T}\int_0^T \sum_{i=1}^K\sum_{j=1}^L} {}\times\frac{X_{ij}(t)}{2}\left[\sum_{k=1}^K\sum_{l=1}^L\left(s_{ij,kl}X_{kl}(t) - \sum_{m=1}^K\sum_{n=1}^Ls_{kl,mn}X_{kl}(t)X_{mn}(t)\right)\right]\dd t\notag\\
= {} & \hat{\rho} - \frac{1}{I_T}\int_0^T \sum_{i=1}^K\sum_{j=1}^L(X_{i\cdot}(t)X_{\cdot j}(t) - X_{ij}(t))\sum_{k=1}^K\sum_{l=1}^L\frac{s_{ij,kl}}{2}X_{kl}(t)\dd t.\label{eq:rho-sel}
\end{align}
The last term, which is linear in the selection parameters, quantifies the error introduced by ignoring selection. However, it demonstrates a remarkable property in the absence of epistasis. In that case we can write $s_{ij,kl} = s^\A_{ik} + s^\B_{jl}$, where $s^\A_{ik}$ is the selection parameter associated with genotype $ik$ at locus A, and similarly for $s^\B_{jl}$. Then equation \eqref{eq:rho-sel} simplifies to
\[
\hat{\rho}_{\text{sel}} = \hat{\rho}.
\]
That is, we have an attractive robustness property: if an investigator uses the incorrect model for selection then the estimator $\hat{\rho}$ is unaffected provided selection is not epistatic. The observed information is also the same. Noting that \thmref{thm:error} continues to hold when $c(\cdot)$ is altered, we conclude that $\hat{\rho}_{\text{MLE}}$ defined in \sref{sec:corrected} is still the MLE for $\rho$ in the presence of (non-epistatic) selection.

\subsection{General confounding}
Returning to the general inference problem of \sref{sec:MLE}, we can generalise the previous observations by asking: when does a contribution ${c}(x)$ to the drift leave the estimator $\hat{\varphi}$ unchanged? From \eqref{likelihood:0}, its contribution to the estimator via $\dd \widetilde{X}(t)$ will be zero if and only if $a(X(t);\varphi)^\top V(X(t))^{-1}c(X(t))$ does not depend explicitly on $\varphi$. When the drift is linear in the parameters as in \eqref{eq:linear}, this requirement becomes
\begin{itemize}
\item[($\ast$)] ${\varphi}^\top Z^\top V^{-1}c$ does not depend on $\varphi$.
\end{itemize}
If ($\ast$) holds we will say that the estimation problem for $\hat{\varphi}$ is \emph{robust} to the contribution of $c(x)$. In the context of the Wright--Fisher diffusion we have the following result.
\begin{proposition}
\label{prop:noconfounding}
For a Wright--Fisher diffusion with drift coefficient $\mu(x;\varphi) = c(x) + Z(x)\varphi$ and diffusion coefficient $V = (V_{ij})$, $V_{ij}(x) = x_i(\gd_{ij}-x_j)$, the estimator $\hat{\varphi}$ in \eqref{eq:varphi-hat} is robust to $c(x)$ if and only if
\begin{equation}
\label{eq:noconfounding}
\sum_{i=1}^d \frac{1}{x_i}\frac{\partial a_i}{\partial \varphi_k}(x) c_i(x) = 0,
\end{equation}
for each $k=1,\dots,r$.
\end{proposition}
\begin{proof}
We determine ($\ast$) for the first $d-1$ coordinates of the Wright--Fisher diffusion, with $[V^*(x)]^{-1}$ as in \eqref{eq:Vinverse}:
\begin{align*}
{\varphi}^\top Z(x)^\top [V^*(x)]^{-1}c(x) &= \sum_{k=1}^r \varphi_k \sum_{i=1}^{d-1}\left(\frac{1}{x_i}\frac{\partial a_i}{\partial \varphi_k}(x) - \frac{1}{x_d}\frac{\partial a_d}{\partial \varphi_k}(x)\right)c_i(x)\\
&= \sum_{k=1}^r \varphi_k \sum_{i=1}^{d}\frac{1}{x_i}\frac{\partial a_i}{\partial \varphi_k}(x)c_i(x), & c_d(x) &:= -\sum_{i=1}^{d-1}c_i(x).
\end{align*}
Since $\frac{\partial a_i}{\partial \varphi_k}(x)$ and $c_i(x)$ do not depend on $\varphi$, the above quantity is a linear combination of the $\varphi_k$. It does not depend on any $\varphi_k$ if and only if each of its coefficients is zero, i.e.~\eqref{eq:noconfounding} holds.
\end{proof}
To give another example of the applicability of \propref{prop:noconfounding}, consider reversing the roles of selection and recombination, so that we are interested in designing an estimator for (non-epistatic) selection at locus A in the confounding presence of recombination. For simplicity we focus on a genic selection model without mutation:
\begin{align*}
c_{ij}(x) &= \rho (x_{i\cdot}x_{\cdot j} - x_{ij}),\\
a_{ij}(x; s^\A_1,\dots,s^\A_K) &= 
\frac{x_{ij}}{2}\left(s_i^\A - \sum_{k=1}^K s_k^\A x_{k\cdot}\right), \qquad i=1,\dots,K;\; j=1,\dots,L.
\end{align*}
In this model we find
\begin{align*}
\sum_{i=1}^K\sum_{j=1}^L \frac{1}{x_{ij}}\frac{\partial a_{ij}}{\partial s^\A_k}(x) c_{ij}(x) &= \sum_{i=1}^K\sum_{j=1}^L \frac{1}{x_{ij}}\frac{x_{ij}}{2}(\gd_{ik} - x_{k\cdot})\rho(x_{i\cdot}x_{\cdot j} - x_{ij}) = 0,
\end{align*}
and \eqref{eq:noconfounding} holds; by \propref{prop:noconfounding} the estimator $\hat{s}^\A$ for $(s^\A_1,\dots,s^\A_K)$ is robust to recombination, as we might hope.

Using \propref{prop:noconfounding} it is also possible to show that the following problems are robust:
\begin{enumerate}[(i)]
\item Estimation of mutation at locus A when there is selection at locus B,
\item Estimation of genic selection at locus A when there is mutation at locus B;
\end{enumerate}
while the following problems are \emph{not} robust:
\begin{enumerate}[(i)]
 \setcounter{enumi}{2}
\item Estimation of recombination when there is mutation at either locus,
\item Estimation of mutation at locus A when there is mutation at locus B,
\item Estimation of mutation at locus A when there is recombination,
\item Estimation of mutation at locus A when there is selection at locus A,
\item Estimation of genic selection at locus A when there is mutation at locus A;
\end{enumerate}
similarly for problems interchanging the two loci. We omit the straightforward calculations. We caution that in the estimation problems above, the likelihood, and thus the estimator $\hat{\varphi}$, will be valid only up to time $S$ as in \eqref{eq:stopping}. It is possible to have $I_T = \infty$ even in models without recombination (in particular, we expect two path measures with different mutation parameters to be mutually singular if certain allele frequencies reach 0).
\subsection{Joint estimation of recombination and selection}
\label{sec:joint}
In contrast to \sref{sec:confounding-sel}, one might recognise the possible existence of selection and be interested in constructing a joint estimator for recombination and selection. How does the marginal estimator for $\rho$ from this compare to those already developed? To illustrate the idea, we consider a simple genic selection model at locus A in which only allele $k$ is under selection; that is, $s^\A_i = 0$ for $i\neq k$:
\begin{align*}
c_{ij}(x) &= \frac{\tA}{2}\sum_{k=1}^K x_{kj}(P_{ki}^\A - \delta_{ik}) + \frac{\tB}{2}\sum_{l=1}^L x_{il}(P_{lj}^\B - \delta_{jl}),\\
a_{ij}(x; \rho, s^\A_k) &= \rho (x_{i\cdot}x_{\cdot j} - x_{ij}) + 
\frac{x_{ij}}{2}\left(\gd_{ik}s_k^\A - s_k^\A x_{k\cdot}\right), \qquad i=1,\dots,K;\; j=1,\dots,L.
\end{align*}
From \eqref{eq:varphi-hat} and \eqref{obsinf} we find
\begin{align*}
I_T &= \begin{pmatrix}
\ds\int_0^T\sum_{i=1}^K\sum_{j=1}^L\frac{(X_{ij}(t)-X_{i\cdot}(t)X_{\cdot j}(t))^2}{X_{ij}(t)}\dd t & 0\\
0 & \ds\frac{1}{4}\int_0^T X_{k\cdot}(t)(1-X_{k\cdot}(t)) \dd t
\end{pmatrix},\\
Y &= \begin{pmatrix} \ds\int_{0}^{T}\sum_{i=1}^K\sum_{j=1}^L\frac{X_{i\cdot}(t)X_{\cdot j}(t)}{X_{ij}(t)}\dd\widetilde{X}_{ij}(t)\\ \ds\frac{1}{2}\int_0^T \dd\widetilde{X}_{k\cdot}(t)\end{pmatrix},
\end{align*}
and so $\hat{\varphi} = I_T^{-1}Y$ simplifies to
\[
\hat{\varphi} = \begin{pmatrix} \hat{\rho} \\
\hat{s}_k^\A \end{pmatrix},
\]
where $\hat{\rho}$ is the same estimator as we found in \eqref{estimate} and
\[
\hat{s}_k^\A = \frac{2(\widetilde{X}(T) - \widetilde{X}(0))}{\int_0^T X_{k\cdot}(t)(1-X_{k\cdot}(t)) \dd t}.
\]
Setting $\tA = \tB = 0$ recovers the estimator for selection found by \citet{wat:1979}, up to a choice of timescale. The key point is that the presence of selection as a `known-unknown' leaves the estimator for $\rho$ unaffected, as is clear from the diagonal nature of $I_T$. In fact this might have been predicted even earlier: requiring $I_{kl} = 0$ for the off-diagonal entry of an observed information matrix, corresponding to two parameters $\varphi_k$, $\varphi_l$, is essentially equivalent to the robustness condition ($\ast$). (To see this we identify the $c_i(x)$ term in ($\ast$) with $Z_{il}\varphi_l$, so $\varphi_l$ parametrises what would have been a confounder in ($\ast$).)

\section{Simulation study}
\label{sec:simulation}
In this section we conduct an empirical study of the properties of $\hat{\rho}_{\text{MLE}}$ by simulation. Although there has been recent progress in the development of algorithms for \emph{exact} simulation of certain classes of Wright--Fisher diffusion \citep{jen:spa:2017, gri:etal:2018, gar:etal:2021}, these algorithms do not cover the non-reversible diffusions considered in this paper. Instead we resort to simple Euler--Maruyama simulation; that is, to simulate small increments of the diffusion over a fixed, small timestep $\gD t$ using the approximation
\[
X(t+\gD t) = X(t) + [c(X(t)) + a(X(t); \varphi)]\gD t + \gs(X(t))[W(t+\gD t) - W(t)], \quad X(0) = x(0),
\]
where $W(t)$ is the $(d-1)$-dimensional Brownian motion in \eqref{eq:SDE}. Integrals involving the sample path of $X$ can be approximated using Riemann sums constructed from the same set of gridpoints.

Because of its singularities at the boundaries of $\gD_{d-1}$, the Cholesky decomposition of $V(x)$ of \citet{sat:1976} is perhaps not the best choice of $\gs(x)$ for the purposes of simulation, a point also noted in \citet{he:etal:2020}. Instead we use a decomposition introduced by \citet{pal:2011, pal:2013}: 
\[
\gs_{\bfi\bfj}(x) = \sqrt{x_\bfi}(\gd_{\bfi\bfj} - \sqrt{x_{\bfi}x_{\bfj}}), \qquad \bfi,\bfj = 1,\dots,d.
\]
This formulation has the advantage of being simple, symmetric, bounded in $x$, and vectorising easily via
\[
\gs(x) = (I_d - \text{diag}(x)1_d)\text{diag}(\sqrt{x}),
\]
where $I_d$ is the identity matrix and $1_d$ is the $d\times d$ matrix of ones. A disadvantage is that it uses a $d$-dimensional Brownian motion, one dimension more than is necessary for simulation.

Using Euler--Maruyama simulation it is possible to obtain a realisation with $X_{ij}(t + \gD t) \leq 0$ for some $(i,j)$. If this occurs we set $X_{ij}(t + \gD t) = 0$, renormalize $X(t+\gD t)$ so that $X(t+\gD t) \in \gD_{d-1}$, and set $I_{t+\gD t} = \infty$.

In the following we posit a two-locus, diallelic ($K=L=2$) model with symmetric mutation ($P^\A = P^\B = \left(\begin{smallmatrix} 1/2\\1/2 \end{smallmatrix}\right)$) and initial condition $X(0) = \left(\begin{smallmatrix} 2/5 & 1/5\\ 1/5 & 1/5\end{smallmatrix}\right)$. The stepsize is set to $\gD t = 10^{-6}$ and paths are simulated up to a time $T = 1$. We consider two sets of mutation parameters: (i) $\tA = \tB = 1$, and (ii) $\tA = \tB = 5$, in order to distinguish models in which the boundaries can or cannot be approached. We explore a variety of recombination parameters, $\rho \in \{0,0.1,1,2.5,5,10,25\}$, and to estimate distributional properties of the estimator we repeat each experiment 100 times.

As an illustration and a check that our implementation is accurate, examples of individual sample paths for $\rho = 5$ are shown in \fref{fig:traj1} and \fref{fig:traj2}, together with the accumulated information, $I_t$, and the evolving error, $\hat{\rho}_{\text{MLE}} - \rho$, as functions of time. As is clear from the Figures, the error is stochastically converging towards 0, with erratic jumps towards 0 in regions where the information accumulates most rapidly. In the second example, in which $\tA=\tB=1$, the trajectory for $X_{22}(t)$ wanders sufficiently closely to 0 that $I_t = \infty$ for some $t < T$, whereupon $\hat{\rho}_{\text{MLE}} = \rho$. The distribution of $\hat{\rho}_{\text{MLE}}$ across 100 experiments using these parameters are shown in \fref{fig:hist}, with results for further experiments summarised in \tref{tab:main}. 

\begin{figure}[p]
\begin{center}
\begin{tabular}{cc}
\includegraphics[width=0.85\textwidth]{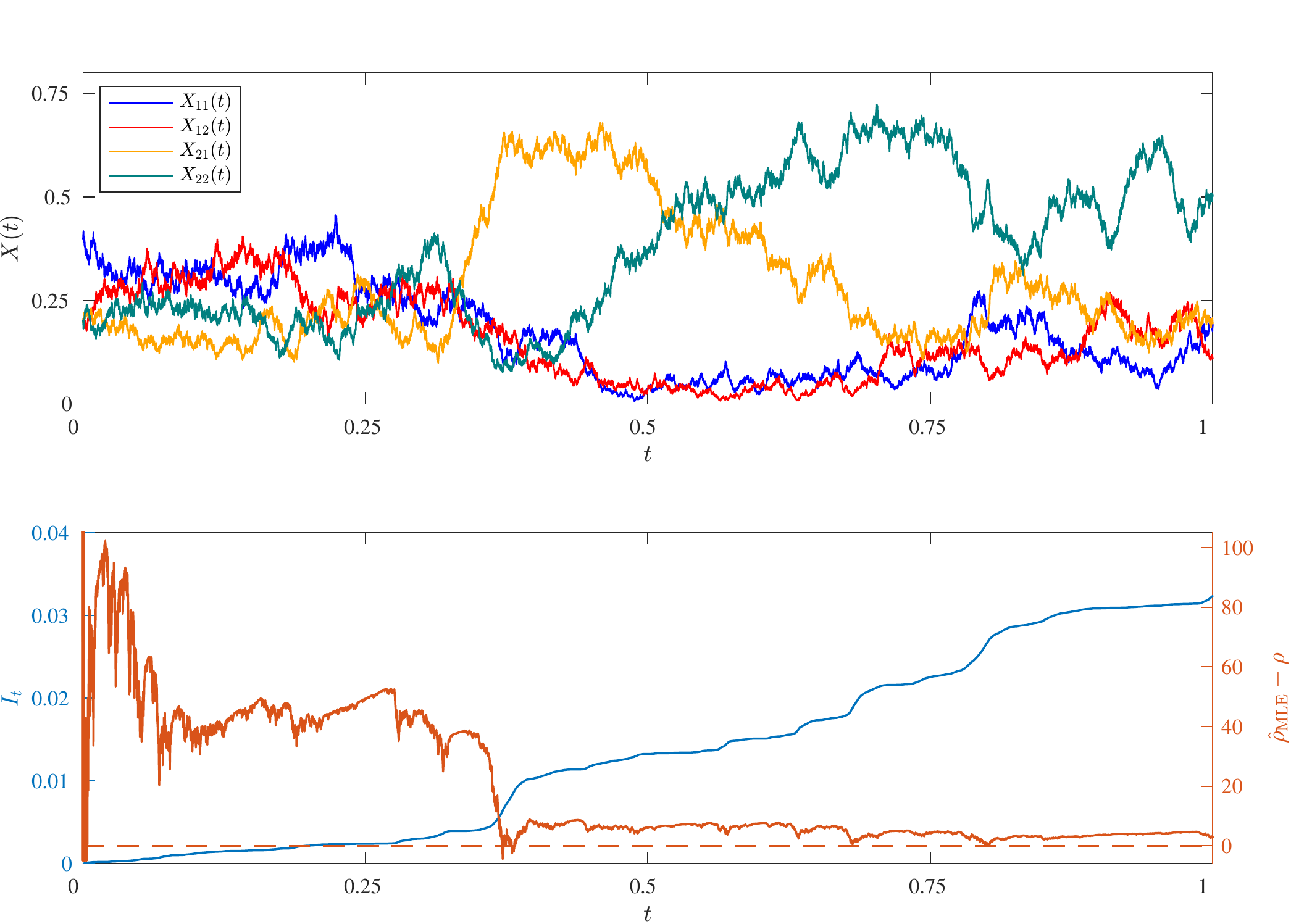} \end{tabular} 
\end{center}
\caption{\label{fig:traj1}Example trajectories in a two-locus model with two alleles at each locus, $\rho = 5$, and $\tA=\tB=5$. Also shown in the lower plots are the trajectories of $\hat{\rho}_{\text{MLE}}-\rho$ and $I_t$ for this sample path.}
\end{figure}

\begin{figure}[p]
\begin{center}
\begin{tabular}{cc}
\includegraphics[width=0.85\textwidth]{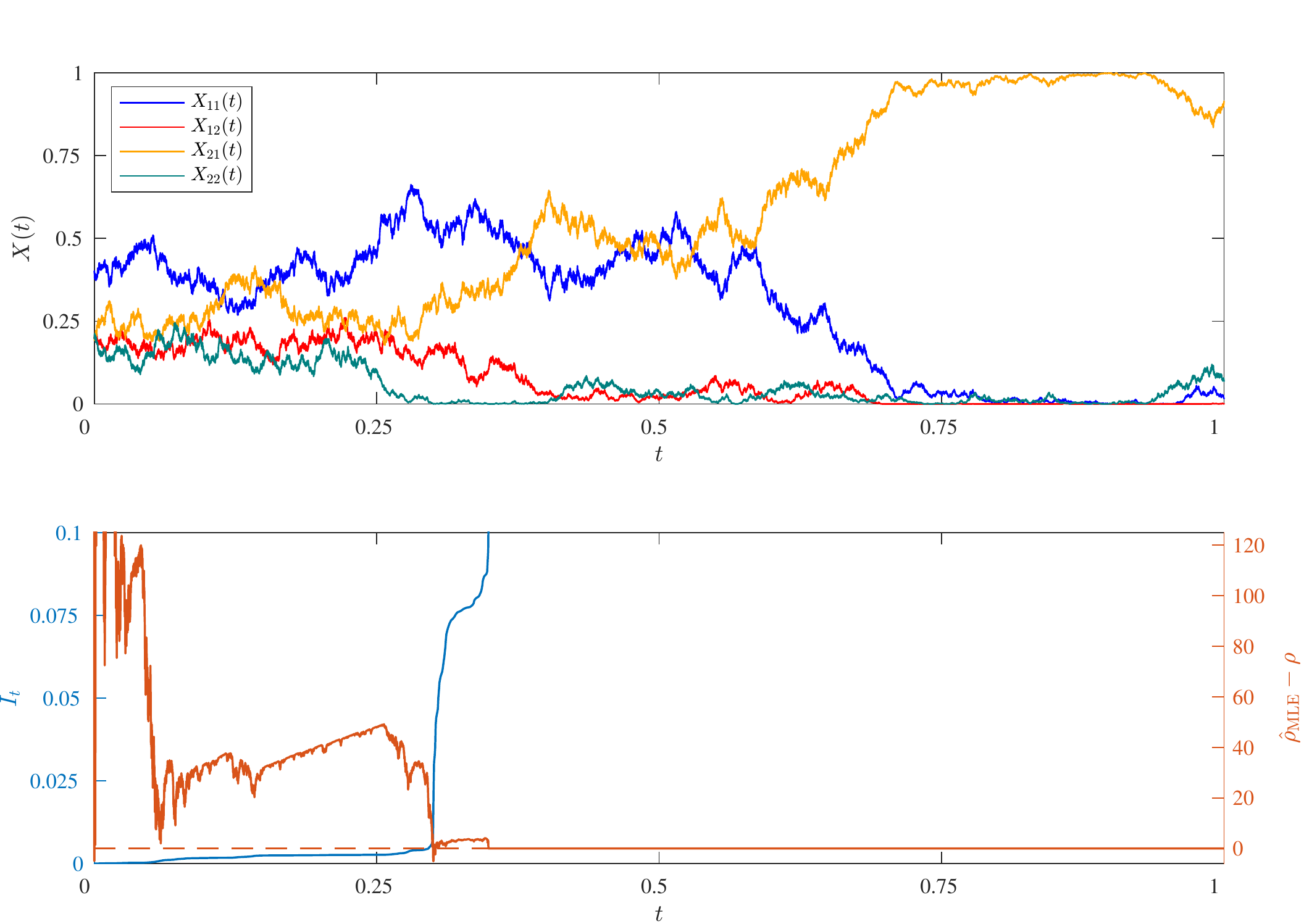} 
\end{tabular}
\end{center}
\caption{\label{fig:traj2}As \fref{fig:traj1} but with $\tA=\tB=1$.}
\end{figure}

\begin{figure}[p]
\begin{center}
\begin{tabular}{cc}
$\tA = \tB = 5$ & $\tA = \tB = 1$\\
\includegraphics[width=0.45\textwidth]{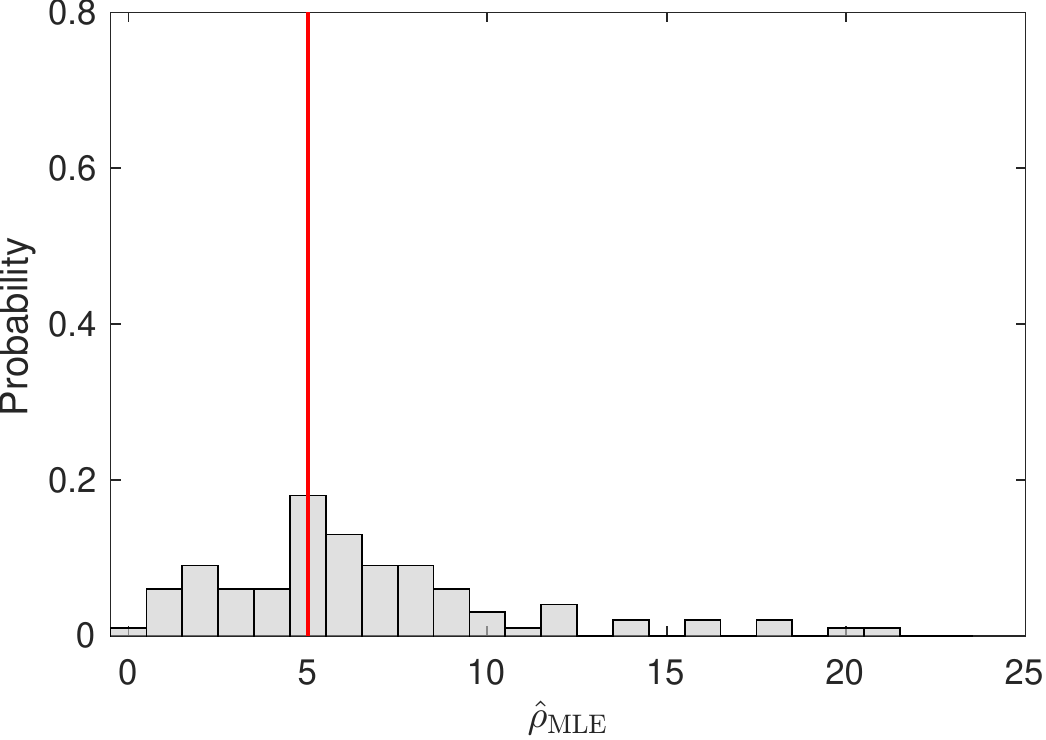} 
& \includegraphics[width=0.45\textwidth]{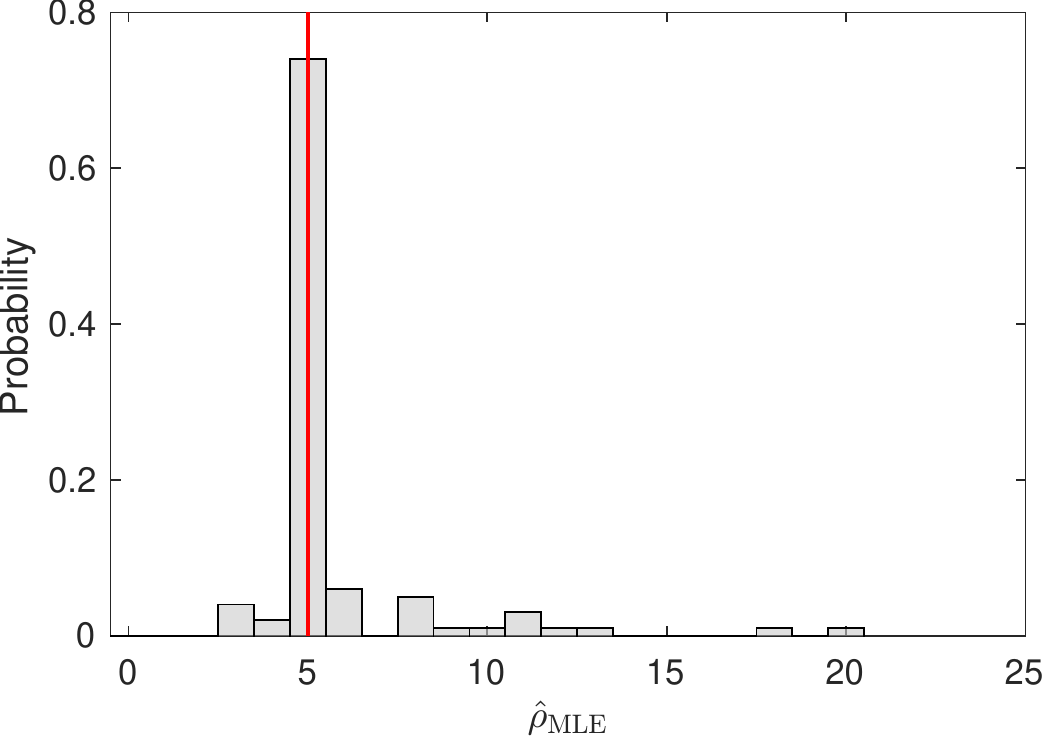} 
\end{tabular}
\end{center}
\caption{\label{fig:hist}Distribution of $\hat{\rho}_{\text{MLE}}$ estimated from 100 replicates. Mutation parameters are $\tA=\tB=5$ (left) and $\tA=\tB=1$ (right). The true recombination parameter, shown by a red line, is $\rho = 5$.}
\end{figure}

\newcolumntype{d}[1]{D{.}{.}{#1} }
\begin{table}[p]
\begin{center}
{\scriptsize
  \begin{tabular}{ d{1} | r@{.}l r@{.}l r@{.}l r@{.}l r@{.}l r@{.}l r@{.}l }
  \hline
    \multicolumn{15}{c}{$\tA = \tB = 5$}\\
    \hline
    \rho & \multicolumn{2}{c}{$\bbE(\hat{\rho}_{\text{MLE}})$} 
     & \multicolumn{2}{c}{$\var(\hat{\rho}_{\text{MLE}})$} & \multicolumn{2}{c}{$\bbE(\hat{\rho}_{\text{MLE},(5)})$} & \multicolumn{2}{c}{$\text{Median}(\hat{\rho}_{\text{MLE}})$} & \multicolumn{2}{c}{$\bbE(\hat{\rho}_{\text{MLE},(95)})$} & \multicolumn{2}{c}{$\bbP(\hat{\rho}_{\text{MLE}} = \rho)$} & \multicolumn{2}{c}{$\bbP(\Lambda > F_m^{-1}(0.95))$}\\ \hline
           0 &       2&69 & 
                  13&52 &            0&00 &       1&06 &       9&81 &           0&11 & 0&10 \\
          0.1 &       2&58 & 
                 13&57 &            0&00 &      0&92 &       9&99 &           0&09 & 0&19\\
            1 &       2&69 &  
                   11&15 &            0&00 &       1&45 &       9&55 &            0&09 & 0&18 \\
          2.5 &       4&75 & 
                 20&25 &            0&00 &       3&36 &       14&98 &            0&08 & 0&32 \\
            5 &       6&74 &  
                   20&18 &       1&41 &        5&89 &       16&76 &            0&07 & 0&38 \\
           10 &       12&44 & 
                  39&94 &       4&26 &       10&94 &        26&49 &            0&04 & 0&79 \\
           25 &       28&38 & 
                  69&79 &       16&89 &       27&36 &       44&70 &            0&02 & 1&00\\
     \hline \multicolumn{15}{c}{}
  \end{tabular}
    \begin{tabular}{ d{1} | r@{.}l r@{.}l r@{.}l r@{.}l r@{.}l r@{.}l r@{.}l }
  \hline
    \multicolumn{15}{c}{$\tA = \tB = 1$}\\
    \hline
    \rho & \multicolumn{2}{c}{$\bbE(\hat{\rho}_{\text{MLE}})$} 
     & \multicolumn{2}{c}{$\var(\hat{\rho}_{\text{MLE}})$} & \multicolumn{2}{c}{$\bbE(\hat{\rho}_{\text{MLE},(5)})$} & \multicolumn{2}{c}{$\text{Median}(\hat{\rho}_{\text{MLE}})$} & \multicolumn{2}{c}{$\bbE(\hat{\rho}_{\text{MLE},(95)})$} & \multicolumn{2}{c}{$\bbP(\hat{\rho}_{\text{MLE}} = \rho)$} & \multicolumn{2}{c}{$\bbP(\Lambda > F_m^{-1}(0.95))$}\\ \hline
             0 &      0&21 & 
                    1&07 &            0&00 &            0&00 &       0&60 &         0&95 & 0&03 \\
          0.1 &      0&34 &  
                 1&61 &          0&10 &          0&10 &       1&07 &         0&95 & 0&98 \\
            1 &       1&22 &   
                    1&38 &            1&00 &            1&00 &       1&96 &         0&95 & 0&98 \\
          2.5 &       3&07 &   
                 2&78 &          2&50 &          2&50 &       6&69 &         0&77 & 0&92 \\
            5 &       5&79 & 
                    6&55 &       4&09 &            5&00 &           11&00 &         0&71 & 1&00 \\
           10 &       10&73 & 
                  4&88 &       9&74 &           10&00 &       17&29 &         0&75 & 1&00 \\
           25 &       26&49 & 
                  27&72 &       20&64 &           25&00 &       36&30 &         0&65  & 1&00\\
    \hline
  \end{tabular}
  }
\caption{\label{tab:main}Distributional summaries of $\hat{\rho}_{\text{MLE}}$ for $\tA = \tB = 5$ (top) and $\tA = \tB = 1$ (bottom): mean, variance, 5th percentile, median, 95th percentile, frequency of zero error, and power to reject $\rho_0 = 0$ at level 5\%. Each estimate is based on 100 independent replicates.}
\end{center}
 \end{table}
 
 As is clear from \tref{tab:main}, $\hat{\rho}_{\text{MLE}}$ is slightly upwardly biased, with the relative bias greater for $\rho$ close to 0. Even with our assumption that the entire sample path is observed, for $\tA=\tB=5$ the distribution of $\hat{\rho}_{\text{MLE}}$ is rather flat: for example, when $\rho=2.5$ the central 90\% of its mass is approximately contained in the interval $[0,15]$. The power to reject the hypothesis $\rho_0 = 0$ at level $5\%$ is consequently poor for small $\rho$, exceeding 0.5 only for the rows in the table with $\rho \geq 10$.
 
 For $\tA = \tB = 1$ the picture is very different, demonstrating the sensitivity of $\hat{\rho}_{\text{MLE}}$ to the mutation parameters. We can see that here there is high probability that $\bbP(\hat{\rho}_{\text{MLE}} = \rho)$, providing very high power to reject $\rho_0 = 0$ even for small $\rho > 0$.

\section{Discussion}
\label{sec:discussion}
In this article we have derived an expression for the maximum likelihood estimator of the recombination rate, $\hat{\rho}_{\text{MLE}}$, from a continuously observed diffusion model of haplotype frequencies. As well as recombination, the diffusion model can incorporate mutation, selection, and genetic drift. We have investigated the empirical properties of the estimator and its robustness to the presence of other processes. We have shown that, contrary to a typical estimator, it is possible to have $\hat{\rho}_{\text{MLE}} = \rho$ with positive probability, and this event is intimately associated with the hitting of the boundary by the diffusion (\thmref{thm:explode}). Although in that theorem we made some convenient assumptions about the trajectory of $X(t)$, we expect it is possible to refine this result further; indeed we conjecture that $\{I_T = \infty\}$ is equal to the event that one haplotype frequency reaches 0 by time $T$. This would provide an easy way to check whether $\{I_T = \infty\}$ has occurred.

Although \thmref{thm:error} and \thmref{thm:explode} are written in statistical language, in terms of estimators and information, we can gain some further intuition by phrasing them in a more fundamental way: it is known that the non-explosion condition \eqref{eq:information} holds if and only if $\bbP_\varphi^{(T)} \ll \bbP_{\varphi_0}^{(T)}$ \citep{hob:rog:1998}. Thus, when estimating recombination (or mutation, but not selection), hitting a boundary of the diffusion leads to the loss of absolute continuity of one path measure with respect to another. The information $I_T$ provides a natural measure of `signal-to-noise'. From the point of view of a finite population, although one usually thinks of stochastic effects as being more important when an allele is very rare compared to when it is common, on the contrary what matters in the diffusion limit here is that the variance in offspring distribution (noise) goes to zero at the boundary while the mean detectable effect of recombination (signal) does not. (The qualitatively different behaviour at a boundary between a finite population model and its diffusion limit is also remarked on by \citet[p180]{ewe:2004:I}.) This also explains the effects of the mutation rate on estimation of $\rho$ as observed in \sref{sec:simulation}: higher mutation rates act to push haplotype frequencies toward the interior of the simplex, where the accumulation of information is slower. It also matches biological intuition: if mutation rates are very small, we can reject a null of no recombination by using the four-gamete test on just a sample at a single time point. As mutation rates increase, it is harder to tell apart recurrent mutation from recombination.

Because of the unusual behaviour of the estimator for $\rho$, we have refrained from providing a detailed description of its asymptotic properties such as local asymptotic normality (LAN). Using the estimator for the immigration rate of the CBI diffusion as a guide, it should be possible to show that inference for $\rho$ exhibits LAN along the sequence of random times
\[
T_n := \inf\{t\in [0,T]:\: I_t = n\};
\]
see \citet[\S 3.4]{ove:1998}. However, in the case $I_T < \infty$ for each $T$, finding the asymptotic behaviour of $\hat{\rho}_{\text{MLE}}$ is more involved since we do not know the stationary distribution of $X$.

Finally, we observe that fundamental quantities appearing throughout this work are
\begin{equation}
\label{eq:newLD}
\sum_{i=1}^K\sum_{j=1}^L \frac{(X_{ij} - X_{i\cdot}X_{\cdot j})^2}{X_{ij}},\qquad \text{and}\qquad
\sum_{i=1}^K\sum_{j=1}^L \frac{X_{ij} - X_{i\cdot}X_{\cdot j} }{X_{ij}}.
\end{equation}
The role of $X_{ij}$ in the denominators has been to upweight the importance of those parts of the trajectories where the frequency of haplotype $(i,j)$ is small, since these regions are more informative for $\rho$. In one sense this is unsatisfactory since the diffusion model is often regarded as an approximation of a discrete population of size $N$, and behaviour near the boundaries is inappropriate when true frequencies can only be a multiple of $1/N$. One might prefer to replace the estimator $\hat{\rho}$, which integrates each $X_{ij}(t)$ over $[0,T]$, with one that integrates $X_{ij}(t)$ only over some sub-region
\[
\{t \in [0,T]:\: \varepsilon \leq X_{ij}(t) \leq 1-\varepsilon\}.
\]
Even under this restriction, the quantities in \eqref{eq:newLD} tell us to focus our attention on those regions where the haplotype frequency is far from $1/2$. The normalizations inherent in \eqref{eq:newLD} seem to offer `natural' new normalizations for the coefficient of linkage disequilibrium, which do not correspond to the usual normalizations found in $r^2$ and $D'$, for example \citep{sve:hil:2018}. Exploring the properties of these new summaries of LD will be the subject of future work.

\section*{Data availability}
Data and code used to conduct the simulation study is available at \url{https://github.com/Paul-Jenkins/GriffithsJenkins2023}.

\section*{Acknowledgements}
This work was supported by The Alan Turing Institute under the EPSRC grant EP/N510129/1. PJ acknowledges the organisers and participants of the Workshop on Genetic Recombination held at Bielefeld University, 08--09 Nov 2022, who provided helpful feedback on an earlier version of this work.

\bibliographystyle{myplainnat}

\appendix
\section{Appendix}
\label{sec:appendix}
In this section we study the deterministic mutation-recombination equation
\begin{equation}
\frac{\dd x_{ij}}{\dd t} = \rho (x_{i\cdot}x_{\cdot j}-x_{ij}) + \frac{\tA}{2}(x_{\cdot j}P_{i}^\A - x_{ij}) + \frac{\tB}{2}(x_{i\cdot}P_{j}^\B - x_{ij}), \qquad i=1,\dots,K;\: j=1,\dots,L.
\label{de2}
\end{equation}
Summing \eqref{de2} over $j$ yields:
\[
\frac{\dd x_{i\cdot}}{\dd t} = \frac{\tA}{2}(P_{i}^\A - x_{i\cdot}), \qquad i=1,\dots,K,
\]
whose solution is
\[
x_{i\cdot}(t) = P^\A_i + (x_{i\cdot}(0) - P^\A_i)e^{-\frac{\tA}{2}t}, \qquad i=1,\dots,K.
\]
Similarly,
\[
x_{\cdot j}(t) = P^\B_j + (x_{\cdot j}(0) - P^\B_j)e^{-\frac{\tB}{2}t}, \qquad j=1,\dots,L.
\]
Therefore, \eqref{de2} can be written
\[
\frac{\dd x_{ij}}{\dd t}  + \left(\rho + \frac{\theta}{2}\right)x_{ij}(t) = \rho x_{i\cdot}(t)x_{\cdot j}(t) + \frac{\tA}{2}x_{\cdot j}(t)P^\A_i + \frac{\tB}{2}x_{i\cdot}(t)P^\B_j =: F(t),
\]
with $F(t)$ known and $\theta := \tA + \tB$. This can be solved via an integrating factor; the solution is
\begin{align*}
x_{ij}(t) = {} & e^{-\left(\rho + \frac{\theta}{2}\right)t}\left[x_{ij}(0) + \int_0^t F(s)e^{\left(\rho + \frac{\theta}{2}\right)s} \dd s\right]\\
= {} & \left(1-e^{-\left(\rho + \frac{\theta}{2}\right)t}\right)P^\A_iP^\B_j\\
& {}+ (x_{\cdot j}(0) - P^\B_j)P^\A_i\left(e^{-\frac{\tB}{2}t} - e^{-\left(\rho + \frac{\theta}{2}\right)t}\right)
 {}+ (x_{i\cdot}(0) - P^\A_i)P^\B_j\left(e^{-\frac{\tA}{2}t} - e^{-\left(\rho + \frac{\theta}{2}\right)t}\right)\\
& {}+ (x_{i\cdot}(0) - P^\A_i)(x_{\cdot j}(0) - P^\B_j)\left(e^{-\frac{\theta}{2}t} - e^{-\left(\rho + \frac{\theta}{2}\right)t}\right)
{}+ x_{ij}(0)e^{-\left(\rho + \frac{\theta}{2}\right)t}.
\end{align*}
Substituting the solutions for $x_{i\cdot}(t)$, $x_{\cdot j}(t)$, and $x_{ij}(t)$ into \eqref{eq:information-recombination} gives an expression for the observed information. After some simplification we arrive at
\[
I_T = \int_0^T \sum_{i=1}^K\sum_{j=1}^L \frac{D_{ij}(0)^2 e^{-2\left(\rho + \frac{\theta}{2}\right)t}}{D_{ij}(0)e^{-\left(\rho + \frac{\theta}{2}\right)t} + [P^\A_i + (x_{i\cdot}(0)-P^\A_i)e^{-\frac{\tA}{2}t}][P^\B_j + (x_{\cdot j}(0)-P^\B_j)e^{-\frac{\tB}{2}t}]} \dd t.
\]
where $D_{ij}(0) = x_{ij}(0)-x_{i\cdot}(0)x_{\cdot j}(0)$. Since the integrand of $I_T$ decays exponentially in $t$, clearly $I_\infty < \infty$ as before.
\end{document}